\documentclass{article}

\usepackage{amsmath}
\usepackage{algorithm}
\usepackage{algcompatible}
\usepackage{amssymb}
\usepackage{mathtools}

\usepackage{tikz}
\usepackage{quiver}

\algnewcommand\algorithmicreturn{\textbf{return}}
\algnewcommand\RETURN{\algorithmicreturn}
\algnewcommand\algorithmicprocedure{\textbf{procedure}}
\algnewcommand\PROCEDURE{\item[\algorithmicprocedure]}%
\algnewcommand\algorithmicendprocedure{\textbf{end procedure}}
\algnewcommand\ENDPROCEDURE{\item[\algorithmicendprocedure]}%
\algnewcommand{\algvar}[1]{{\text{\ttfamily\detokenize{#1}}}}
\algnewcommand{\algarg}[1]{{\text{\ttfamily\itshape\detokenize{#1}}}}
\algnewcommand{\algproc}[1]{{\text{\ttfamily\detokenize{#1}}}}
\algnewcommand{\algassign}{\leftarrow}

\newtheorem{idea}{Idea}
\newtheorem{theorem}{Theorem}[section]
\newtheorem{definition}{Definition}[section]
\newtheorem{corollary}{Corollary}[section]
\newtheorem{proposition}{Proposition}[section]
\newtheorem{lemma}{Lemma}[section]
\newtheorem{observation}{Observation}[section]
\newtheorem{example}{Example}[section]
\newcommand{\gomdci}{General-OMDCI~}
\newcommand{\omdci}{OMDCI~}
\newcommand{\omdciplus}{OMDCI+~}

\newcommand{\qed}{\hfill\hbox{\rlap{$\sqcap$}$\sqcup$}}
\newenvironment{proof}{\noindent \emph{Proof.\,}}{\qed}

\title{On the difficulty of order constrained pattern matching with applications to feature matching based malware detection.}

\author{Adiesha Liyanage\footnote{Gianforte School of Computing, Montana State University, Bozeman, MT 59717, USA. Email: {\tt a.liyanaralalage@montana.edu}}
\and
Braeden Sopp\footnote{Gianforte School of Computing, Montana State University, Bozeman, MT 59717, USA. Email: {\tt braeden.sopp@student.montana.edu}}
\and 
Binhai Zhu\footnote{Gianforte School of Computing, Montana State University, Bozeman, MT 59717, USA. Email: {\tt bhz@montana.edu}}
}

\date{}

\begin{document}

\maketitle

\begin{abstract}
Cybercrime, primarily in the form of malware, imposes substantial financial burden on its victims, with projections estimating its cost reaching \$10 trillion by 2025.
Malware detection in low-level programs using static analysis began as early as 2003, and feature extraction remains the dominant approach today.
This research enabled subsequent signature-based and AI-based detection algorithms, all of which rely on feature matching in various ways.
However, no theoretical analysis has been conducted on malware detection using algorithms based on features matching.
In this paper, we examine malware detection through feature matching by formalizing the problem based on two observations from reality: (1) Most malware samples received by a single company belong to the same family or subfamily, exhibiting similarities in the low-level program structure, which consists of essentially disjoint blocks (each with a single entry and exit point); (2) Certain critical instruction sequences---such as opening a file followed by write operations, possibly preceded by API calls---are indicative of malware but not a guarantee.

We formulate the problem of low-level malware detection using algorithms based on feature matching as Order-based Malware Detection with Critical Instructions (General-OMDCI): 
given a pattern in the form of a sequence \(M\) of colored blocks, where each block contains a critical character (representing a unique sequence of critical instructions potentially associated with malware but without certainty), and a program \(A\), represented as a sequence of \(n\) colored blocks with critical characters, the goal is to find two subsequences, \(M'\) of \(M\) and \(A'\) of \(A\), with blocks matching in color and whose critical characters form a permutation of each other.
When $M$ is a permutation in both colors and critical characters the problem is called OMDCI. If we additionally require $M'=M$, then the problem is called OMDCI+; if in this case $d=|M|$ is used as a parameter, then the OMDCI+ problem is easily shown to be FPT. 
Our main (negative)
results are on the cases when $|M|$ is arbitrary and are summarized as follows:
\begin{enumerate}
    \item OMDCI+ is NP-complete, which implies OMDCI is also NP-complete.
    \item For the special case of OMDCI, deciding if the optimal solution has length $0$ (i.e., deciding if no part of \(M\) appears in \(A\)) is co-NP-hard. As a result, the OMDCI problem does not admit an FPT algorithm unless P=co-NP.
\end{enumerate}
In summary, our results imply that using algorithms based on feature matching to identify malware or determine the absence of malware in a given low-level program are both hard.


\end{abstract}


\section{Introduction}

Cybercrime, primarily carried out through malware attacks, occurs daily. 
Projections in 2020 indicated that its cost is expected to grow at an annual rate of 15\%, reaching an estimated \$10.5 trillion, up from \$3 trillion in 2015 \cite{Morgan2020}.
This cost exceeds the annual GDP of every country except the United States and China. Therefore, substantial research is necessary to prevent or at least mitigate malware spread.
Notably, Christodorescu and Jha conducted the first research on extracting malware features using static analysis in low-level programs as early as 2003 \cite{DBLP:conf/uss/ChristodorescuJ03}.

Currently, feature extraction algorithms (and matching algorithms) are among the most widely used techniques for malware detection and analysis \cite{Aonzo23,Bonett18,Gagnon17,Hu13,Jindal19,Zhong24a}.
These methods involve static analysis, dynamic analysis, or a combination of both to extract features such as API calls and control flow graphs. Static analysis requires a deep understanding of malware structure, whereas dynamic analysis involves executing the program to traverse most, if not all, of its execution paths.

On top of the features extracted, malware classification can be done using signature-based \cite{Coscia23,Seo23,Vu23,Wu23} and AI-driven algorithms, both methods are based on feature matching in different ways. 
Signature-based algorithms identify similar signature features in potential malware \cite{LiJia23}, while AI-driven algorithms learn these features automatically from prior analysis using known features from training data. 
The latter certainly depends on the ability to automatically learn features from the corresponding AI packages (e.g., a neural network) \cite{Lucas24}. In fact, data visualization has also become a viable solution for malware detection \cite{Jin20}. (As AI-based and visualization-based methods are not particularly relevant to this paper, we refer the reader to the following two papers for more details \cite{Jin20,Lucas24}.)

Recently, Zhong et al. conducted a systematic study to classify
malware into five families and 40 subfamilies \cite{Zhong24}. The five families are: {\tt adware}, {\tt worm}, {\tt trojan}, {\tt pua}, and {\tt virus}. The majority of subfamilies are under {\tt trojan} which has 21 subfamilies.
The basic idea of the research by Zhong et al. is to decompose a large program into {\em blocks}, each of which is a segment of the program with one entry point and one exit point, then applying machine learning methods to determine the similarity of blocks
to an existing dataset.

Although there has been extensive research on malware detection and classification from a practical side, less attention has been given to the algorithmic and computational complexity perspective. How hard is it to detect malware, even approximately? Can malware identified in one architecture B be used to detect malware in another architecture C? If so, what is the computational complexity of this process? As a starting point, we investigate some of these problems in this paper.

Note that many malware are caused by system-based critical
instructions, possibly preceded by API calls. Consider, for example, a {\tt potentially unwanted application (pua)}. A {\tt pua} is often produced by the remote injection of instructions into a local program, which then creates and opens a local file, writes unwanted contents into it, and ultimately forges and executes an unwanted application. In this process, the unique sequence of instructions \(\langle{\tt create}, {\tt open}, {\tt write}\rangle\) are critical instructions that collectively build a {\tt pua}. Note that this sequence can be denoted as a single letter/character.
However, the mere presence of these instructions does not necessarily imply the existence of {\tt pua}.

The general idea of our approach is based on these findings as well as the recent classification method by Zhong et al.; moreover, we follow an interesting finding by Ugarte-Pedrero et al., i.e., more than 70\% of malware samples caught at a company belong to a known family or subfamily \cite{DBLP:journals/tissec/Ugarte-PedreroG18}.
Consequently, identifying new malware based on its similarity to (a superset containing) other malware, as well as the sequence of critical instructions in them, is the key to this paper.

We formulate a low-level program at an assembly language level as a sequence of blocks (each with one entry and exit point); moreover, the blocks have some similarity scores with existing programs written under the same architecture (i.e., with identical assembly instructions) --- this can be done using existing AI packages \cite{Zhong24}. Then, we give each block a unique color that corresponds to the score. Once this is done, within each block we identify the
sequence of critical instructions, possibly preceded by some API calls, that are most likely to cause malware. (Note that, to avoid being caught easily, a malware creator is unlikely to inject more than one such sequence causing malware in one program, let alone in one block.) This can be done by matching the new program with the existing datasets or database; for instance, in \cite{Zhong24} Zhong et al. collected 25739 samples among all the 40 subfamilies. If a block contains no such sequence
of critical instructions, we could just ignore the block.

The formal formulation is the \gomdci problem: 
given a program $A$, as a sequence of colored
blocks each with a critical letter/character (representing the corresponding sequence of critical instructions), we match it with a known pattern $M$ by finding a common subsequence of blocks, in terms of the blocks colors, where both subsequences use the same critical characters an equal number of times. 
(When the context is clear, we simply refer to it as a character.) 
We consider a restricted variant that requires $M$ to be a permutation (in both colors and characters) and appears completely in $A$, which we refer to as \omdciplus. (The intuitive idea is that
if some part of $M$ could cause malware but we are not so sure, then we would like to find the whole $M$ in the input program $A$, without missing any part of it.) We also make another restriction
on the \gomdci problem by requiring $M$ to be a permutation (in both colors and characters); the resulting problem is called \omdci. Clearly, any hardness result on \omdci would also hold on \gomdci.

Our main results are two reductions from the NP-complete problems X3C and Hamiltonian Cycle (really, its complement), which were first shown to be NP-complete by Karp \cite{Karp72}. The first one, from X3C, shows that \omdciplus is NP-complete. The second one is more surprising: we show that, by a reduction from the Complement Hamiltonian Cycle problem (co-HC), deciding
if \omdci has an optimal solution of size zero (i.e., finding if there is no malware in the input program) is
co-NP-hard. The reductions, especially the second one, use new techniques to encode a permutation into a sequence, which might find applications in other problems involving sequences. We assume that readers are familiar with the NP-completeness and FPT concepts. For NP-completeness, we refer the reader to the original paper by Cook \cite{Cook71} and certainly to the classic textbook by Garey and Johnson \cite{DBLP:books/fm/GareyJ79}. For the concepts and properties of FPT, we refer the reader to the textbooks \cite{DF99,FG06}.

The paper is organized as follows. In Section 2, we present necessary definitions. In Section 3, we prove the NP-completeness of \omdciplus. In Section 4, we prove the co-NP-hardness of 0-OMDCI. At the end of both cases, we present some related negative results. We conclude the paper in Section 5.

\section{Preliminaries}
Let $n,i$ and $j$ be natural numbers. We define $[n]:=\{1,2,...,n\}$ and $[i,j]:=\{i,i+1,...,j\}$ for $i<j$. Given an alphabet \(\Sigma\), a string \(S\) over alphabet \(\Sigma\) of length \(n\) is defined as a function:
\[S: [n] \rightarrow \Sigma.\]
A subsequence \(S'\) of \(S\) of length \(k\) is defined by a strictly increasing sequence of indices:
\[I_{S'} = \{i_1, i_2, i_3, \ldots, i_k\}, \text{ where } 1 \leq i_1 < i_2 < \cdots < i_k,\]
such that
\[S'(j) = S(i_j), \quad \forall j \in [k].\]

We use \omdci to denote the problem Order-Based Malware Detection Using Critical Instructions. Given two alphabets \(\Sigma\) and \(\Gamma\), we define a \textbf{string with colored indices} \(R\) as a 3-tuple \(([k], \gamma, \sigma)\) where \(k \in \mathbb{N}\),
\[\gamma : [k] \rightarrow \Gamma, \text{ and } \sigma : [k] \rightarrow \Sigma.\]

For clarity and ease of comprehension, we denote a $3$-tuple $R$ as:
\[R \coloneq \frac{\sigma(1)}{\gamma(1)} \cdot \frac{\sigma(2)}{\gamma(2)} \cdots \frac{\sigma(k)}{\gamma(k)}.\]
We use \(\cdot\) as well as \(\circ\) to denote concatenations of strings with colored indices.
\begin{definition}\label{def:gomdic}
\gomdci problem

Given two alphabets \(\Sigma\) and \(\Gamma\), an input to the \gomdci problem is a pair of strings with colored intervals \((M, A)\) with
\[
M = ([m], \gamma_M, \sigma_M), \quad A = ([n], \gamma_A, \sigma_A).
\]
Here, \(\gamma_M\) and \(\gamma_A\) are strings over the alphabet \(\Gamma\). Similarly \(\sigma_M\) and \(\sigma_A\) are strings over \(\Sigma\). 
The decision version of this problem is to determine whether there is a common subsequence of length \(k>0\) between \(\gamma_M\) and \(\gamma_A\) defined over sequences of indices \(I_{M'} = \{a_1, a_2, \ldots, a_k\}\) and 
\(I_{A'}= \{b_1, b_2, \ldots ,b_k\}\) denoted as 
\(M'\) and \(A'\), respectively, with the following properties: 
\begin{enumerate}
    \item \(\gamma_M(a_l) = \gamma_A(b_l), ~~\forall l\in [k].\)
    \label{property:colorsaremakeacommonsubsequence}
    \item \(\mathcal{C}_{M'} = \mathcal{C}_{A'}\) where 
    \[\mathcal{C}_{M'} = \{\sigma_M(i) : i \in I_{M'}\}, \quad \mathcal{C}_{A'} = \{\sigma_A(i) : i \in I_{A'}\}.\]\label{property:setofcoveredcriticalinstructionsareequal}
\end{enumerate}
\end{definition}
In general for a subsequence \(\gamma'\) defined by sequence of indices \(\{i_1,...,i_{l}\}\) of the string \(\gamma\) in \((S, \gamma, \sigma)\), we denote the \textbf{multiset of critical instructions}, \(\{\sigma(i_1), \ldots, \sigma(i_l)\}\), in \(\gamma'\) as \(\mathcal{C}_{\gamma'}\).

\begin{example}
Let $\Gamma=\{1,2,3,4\}$, $\Sigma=\{a_1,a_2,a_3,a_4,a_5\}$, $m=4$ and $n=6$.
       \[ M =\frac{a_1}{1}\cdot \frac{a_2}{2}\cdot \frac{a_4}{1}\cdot \frac{a_3}{4} , ~{and}~       
        A = \frac{a_2}{1}\cdot \frac{a_4}{4}\cdot \frac{a_1}{2}\cdot\frac{a_3}{3}\cdot \frac{a_5}{4} \cdot\frac{a_3}{3} .\]
        The optimal solution is
        \[M'=\frac{a_1}{1}\cdot\frac{a_2}{2} , ~{and}~ A' = \frac{a_2}{1}\cdot \frac{a_1}{2} .\]    
\end{example}
Intuitively, in the \gomdci problem, we would like to find if
any part of some potential malware $M$ (based on prior analyses or AI) appears in a piece of program $A$.
Here, we define two variants of the \gomdci problem, which eventually helps us obtain the hardness results.

\begin{definition}
    \omdciplus problem

\omdciplus is a version of the \gomdci version with following restrictions on the input \(M\) and the solution:
\begin{enumerate}
    \item The strings \(\gamma_M\) and \(\sigma_M\) are \textbf{permutation strings} over the alphabets \(\Gamma\) and \(\Sigma\), respectively.
    \item The decision version of the problem requires determining whether there is a solution such that \(A' = \gamma_M\).
\end{enumerate}
Essentially, the goal is to answer whether the permutation string \(\gamma_M\) of colors appears as a subsequence in \(\gamma_A\) while covering the set of critical instructions associated with the triple \(M\).
\end{definition}
\begin{example}
   Let $\Gamma=\{1,2,3,4\}$, $\Sigma=\{a_1,a_2,a_3,a_4\}$, $m=4$ and $n=6$. An example for OMDCI+ is as follows:
        \[ M =\frac{a_2}{1}\cdot \frac{a_1}{2}\cdot \frac{a_4}{4}\cdot \frac{a_3}{3} , ~{and}~       
        A = \frac{a_1}{1}\cdot \frac{a_4}{4}\cdot \frac{a_2}{2}\cdot\frac{a_3}{3}\cdot \frac{a_3}{4} \cdot \frac{a_4}{3}.\]
        The solution is
        \[A' = \frac{a_1}{1}\cdot \frac{a_2}{2}\cdot\frac{a_3}{4}\cdot \frac{a_4}{3} .\]
\end{example}
Again, intuitively, imagine that we know some blocks in $M$ cause malware through a critical instruction sequence (critical characters in our abstraction) but we are unable to detect which ones. An approach to determining whether some other program $A$ also causes malware would be to determine if $M$ appears in $A$ --- without any element missing. Alternatively, we could instead search for a subset of blocks in $M$ that match blocks in $A$, which is considered in the following restricted variant: 
\begin{definition}\label{def:omdci}
    \omdci problem

     We add the following restrictions to the input triple \(M\) while retaining the original solution restraints from the \gomdci problem.
     \begin{enumerate}
            \item The strings \(\gamma_M\) and \(\sigma_M\) are permutation strings over alphabets \(\Gamma\) and \(\Sigma\).
        \end{enumerate}
\begin{example}
   Let $\Gamma=\{1,2,3,4\}$, $\Sigma=\{a_1,a_2,a_3,a_4\}$, $m=4$ and $n=6$. An example for OMDCI is as follows:
        \[ M =\frac{a_2}{1}\cdot \frac{a_1}{2}\cdot \frac{a_4}{4}\cdot \frac{a_3}{3} , ~{and}~       
        A = \frac{a_4}{1}\cdot \frac{a_1}{4}\cdot \frac{a_2}{2}\cdot\frac{a_4}{3}\cdot \frac{a_3}{4} \cdot \frac{a_2}{3}.\]
        There is no solution for this \omdci instance. In fact, we call this special case as 0-OMDCI (i.e., deciding if OMDCI has a $0$-solution). Note that \omdci is a special case of General-OMDCI, hence any hardness result on it would imply that \gomdci has the same hardness result.
\end{example}
\end{definition}

\section{NP-completeness of OMDCI+}

In this section, as a warm-up, we show that OMDCI+ is NP-complete. Recall that in this case we need to find
the whole (permutation) pattern $M$ in $A$.

\begin{theorem}
The decision version of OMDCI+ is NP-complete.
\label{thm1}
\end{theorem}

\begin{proof}
The NP membership checking is easy: when a possible solution $A'$ is given, we first check if the block sequence match that of $M$ and then check if the set of critical characters is equal to that of $M$. If both tests pass, then return yes; otherwise, return no.

We then reduce Exact-Cover by 3-Sets (X3C) to OMDCI+. The input to X3C is a base set $B=\{b_1,b_2,...,b_{3q}\}$ of $3q$ elements (with $q$ being a positive integer) and a set $S$ 
of $m$ 3-sets in the form of $S=\{S_i|S_i=\{t_{i,1},t_{i,2},t_{i,3}\}\subset B, i\in [m]\}$. The problem is
to decide if $S$ contains a subset $T$ of $q$ 3-sets whose
union is exactly $B$.

For each $j\in [m]$, we construct three intervals with colors $v_{j,1},v_{j,2},v_{j,3}$ which can further become distinct
intervals by adding different superscripts $i\in [q]$. The elements in the base set $B$ would be used as distinct (critical) characters. In addition, we also construct $3m-3$ different characters $C_1,C_2,C_3,...,C_{3m-5},C_{3m-4},C_{3m-3}$, which can also become distinct by adding distinct superscripts.

We define $T^{(i)}$, with $i\in[q]$, which is related to the three elements $b_{3i-2}, b_{3i-1}$ and $b_{3i}$ in a base set $S_i$, as follows.

\[ T^{(i)} = \frac{C^i_{1}}{v^i_{1,1}}\cdot \frac{C^i_{2}}{v^i_{1,2}}\cdot \frac{C^i_{3}}{v^i_{1,3}}
\circ \frac{C^i_{4}}{v^i_{2,1}}\cdot \frac{C^i_{5}}{v^i_{2,2}}\cdot\frac{C^i_{6}}{v^i_{2,3}} \circ \frac{C^i_{7}}{v^i_{3,1}}\cdot \frac{C^i_{8}}{v^i_{3,2}}\cdot\frac{C^i_{9}}{v^i_{3,3}}\circ \cdots\cdots \]
\[~~~~~~~\circ\frac{C^i_{3m-5}}{v^i_{m-1,1}}\cdot \frac{C^i_{3m-4}}{v^i_{m-1,2}} \cdot\frac{C^i_{3m-3}}{v^i_{m-1,3}}
\circ\frac{b_{3i-2}}{v^i_{m,1}}\cdot \frac{b_{3i-1}}{v^i_{m,2}} \cdot\frac{b_{3i}}{v^i_{m,3}} .\]

Note that $T^{(i)}$ has a length of $3m$. We now construct the permutation $M$, with length $3mq$, as
\[ M = T^{(1)}T^{(2)}\cdots T^{(q)} .\]

We next construct the sequence $A$ which could contain $M$ as a subsequence. First, recall that
\[S=\{S_j|S_j=\{t_{j,1},t_{j,2},t_{j,3}\}\subset B, j\in [m]\} .\]
We then define $P^{(i)}_j$, with $i\in [q]$ and $j\in[m]$, which is related to $S_j=\{t_{j,1},t_{j,2},t_{j,3}\}$ as
\[P^{(i)}_j=\frac{t_{j,1}}{v^i_{j,1}}\cdot \frac{t_{j,2}}{v^i_{j,2}} \cdot \frac{t_{j,3}}{v^i_{j,3}},\]
then for $i\in [q]$ and $j\in [m]$
\[H^{(i)}_j=\frac{C^i_{3j-2}}{v^i_{j,1}}\cdot \frac{C^i_{3j-1}}{v^i_{j,2}} \cdot \frac{C^i_{3j}}{v^i_{j,3}} ,\]
and for $i\in[q]$ and $j\in[2,m]$
\[L^{(i)}_j=\frac{C^i_{3j-5}}{v^i_{j,1}}\cdot \frac{C^i_{3j-4}}{v^i_{j,2}} \cdot \frac{C^i_{3j-3}}{v^i_{j,3}} .\]

With these, we define for $i\in [q]$ the sequence $X^{(i)}$, which is of $9(m-1)+3=9m-6$ intervals and colors as follows:
\[ X^{(i)}=P^{(i)}_1L^{(i)}_2H^{(i)}_1\cdot P^{(i)}_2L^{(i)}_3H^{(i)}_2 \cdots \cdots
P^{(i)}_{m-1}L^{(i)}_{m}H^{(i)}_{m-1}\cdot P^{(i)}_{m} .\]

Finally, we construct the string $A$, which is of length
$(9m-6)q$ as follows:

\[ A = X^{(1)}X^{(2)}\cdots X^{(q)} .\]

Since $|M|=3mq$ and $|A|=(9m-6)q$, the reduction obviously takes polynomial time. To close the proof, we claim that X3C has a solution if and only if the intervals in $M$ appears as a subsequence in $A$ and the characters in $M$ is a permutation of that of the characters of the matched intervals in $A$.

The ``only if'' part is proved as follows. If X3C has a solution
$\{Z_1,Z_2,...,Z_q\}$, for each $i\in [q]$ we just take a 
substring of three intervals $P^{(i)}_j$ in $X^{(i)}$ such that the colors in them equals to those in the $Z_i$. After $P^{(i)}_j$ is selected, first note that the colors of intervals in $L^{(i)}_{k+1}$ and $H^{(i)}_k$ are the same, while the corresponding first indexes (subscripts) of the corresponding intervals in $L^{(i)}_{k+1}$ are one more than those in $H^{(i)}_k$, for $k\in [m-1]$. Then, intervals of $C$-colors in $X^{(i)}$ can be selected to match those in $T^{(i)}$ in the same order of $C$-colors, while the indexes of the corresponding intervals (including those in $P^{(i)}_j$) match exactly those in $T^{(i)}$. Clearly, the selected intervals in all $X^{(i)}$'s meet the requirement; or, the instance $\langle M,A\rangle$ has a solution for the OMDCI problem.

For the ``if'' part, we first assume that $\langle M,A\rangle$ has a solution for the OMDCI problem. Due to that
the superscripts of $v$-intervals (and $C$-colors) in $X^{(i)}$ and $T^{(i)}$ are the same (i.e., all are $i$), the selected $v$-intervals in $A$ to match those in $T^{(i)}$ must come from $X^{(i)}$. Further, since $T^{(i)}$ is a permutation, exactly $3(m-1)$ distinct 
$C$-colors with superscripts $i$ must be selected and the three extra colors must be in $B$ --- by construction these three must correspond to a set in $S$ and the union of these colors in all $T^{(i)}$'s must be equal to $B$. Since this holds for all $i\in [q]$, the three non-$C$-colors selected in all $X^{(i)}$'s must cover $B$ exactly. In other words, we have a solution for X3C. 
\end{proof}                          

We show an example to help readers understand the proof better. The X3C instance is:
$q=2$, $B=\{1,2,3,4,5,6\}$, and $S=\{S_1,S_2,S_3,S_4\}$ with $S_1=\{1,3,5\}$, $S_2=\{2,5,6\}$, $S_3=\{2,4,6\}$ and $S_4=\{1,2,4\}$. We have

\[M=T^{(1)}T^{(2)} ,\]
with
\[ T^{(1)} = \frac{C^1_{1}}{v^1_{1,1}}\cdot \frac{C^1_{2}}{v^1_{1,2}}\cdot \frac{C^1_{3}}{v^1_{1,3}}
\circ \frac{C^1_{4}}{v^1_{2,1}}\cdot \frac{C^1_{5}}{v^1_{2,2}}\cdot\frac{C^1_{6}}{v^1_{2,3}} \circ \frac{C^1_{7}}{v^1_{3,1}}\cdot \frac{C^1_{8}}{v^1_{3,2}}\cdot\frac{C^1_{9}}{v^1_{3,3}}
\circ\frac{1}{v^1_{4,1}}\cdot \frac{2}{v^1_{4,2}} \cdot\frac{3}{v^1_{4,3}} ,\]
and
\[ T^{(2)} = \frac{C^2_{1}}{v^2_{1,1}}\cdot \frac{C^2_{2}}{v^2_{1,2}}\cdot \frac{C^2_{3}}{v^2_{1,3}}
\circ \frac{C^2_{4}}{v^2_{2,1}}\cdot \frac{C^2_{5}}{v^2_{2,2}}\cdot\frac{C^2_{6}}{v^2_{2,3}} \circ \frac{C^2_{7}}{v^2_{3,1}}\cdot \frac{C^2_{8}}{v^2_{3,2}}\cdot\frac{C^2_{9}}{v^2_{3,3}}
\circ\frac{4}{v^2_{4,1}}\cdot \frac{5}{v^2_{4,2}} \cdot\frac{6}{v^2_{4,3}} .\]
We also have 
\[A=X^{(1)}X^{(2)} ,\]
with
\[X^{(1)}= \frac{1}{v^1_{1,1}}\cdot \frac{3}{v^1_{1,2}}\cdot \frac{5}{v^1_{1,3}}
\circ \frac{C^1_{1}}{v^1_{2,1}}\cdot \frac{C^1_{2}}{v^1_{2,2}}\cdot\frac{C^1_{3}}{v^1_{2,3}} 
\circ \frac{C^1_{1}}{v^1_{1,1}}\cdot \frac{C^1_{2}}{v^1_{1,2}}\cdot\frac{C^1_{3}}{v^1_{1,3}} 
\circ \frac{2}{v^1_{2,1}}\cdot \frac{5}{v^1_{2,2}}\cdot \frac{6}{v^1_{2,3}}\circ \]
\[~~~~~~~~~\frac{C^1_{4}}{v^1_{3,1}}\cdot \frac{C^1_{5}}{v^1_{3,2}}\cdot \frac{C^1_{6}}{v^1_{3,3}}
\circ \frac{C^1_{4}}{v^1_{2,1}}\cdot \frac{C^1_{5}}{v^1_{2,2}}\cdot \frac{C^1_{6}}{v^1_{2,3}}
\circ \frac{2}{v^1_{3,1}}\cdot \frac{4}{v^1_{3,2}}\cdot \frac{6}{v^1_{3,3}}
\circ \frac{C^1_{7}}{v^1_{4,1}}\cdot \frac{C^1_{8}}{v^1_{4,2}}\cdot \frac{C^1_{9}}{v^1_{4,3}} \circ
\]
\[~~~~~~~~~
\frac{C^1_{7}}{v^1_{3,1}}\cdot \frac{C^1_{8}}{v^1_{3,2}}\cdot \frac{C^1_{9}}{v^1_{3,3}} \circ
\frac{1}{v^1_{4,1}}\cdot \frac{2}{v^1_{4,2}}\cdot \frac{4}{v^1_{4,3}}.~~~~~~~~~~~~~~~~~~~
\]

The solution is $Y^{(1)}Y^{(2)}$, where $Y^{(1)}$ is a
subsequence of $X^{(1)}$, covering the 3-set $S_1=\{1,3,5\}$. To be more
precise,
\[Y^{(1)}= \frac{1}{v^1_{1,1}}\cdot \frac{3}{v^1_{1,2}}\cdot \frac{5}{v^1_{1,3}}
\circ \frac{C^1_{1}}{v^1_{2,1}}\cdot \frac{C^1_{2}}{v^1_{2,2}}\cdot\frac{C^1_{3}}{v^1_{2,3}} 
\circ \frac{C^1_{4}}{v^1_{3,1}}\cdot \frac{C^1_{5}}{v^1_{3,2}}\cdot \frac{C^1_{6}}{v^1_{3,3}}
\circ
\frac{C^1_{7}}{v^1_{4,1}}\cdot \frac{C^1_{8}}{v^1_{4,2}}\cdot \frac{C^1_{9}}{v^1_{4,3}}.
\]
We leave the construction of $X^{(2)}$ and $Y^{(2)}$, which should cover the 3-set $S_3=\{2,4,6\}$, as an exercise for the readers. 


Theorem~\ref{thm1} implies that OMDCI (and \gomdci) are NP-hard. However, for the \omdciplus problem, if $|M|=d$ is a (small) parameter, then the problem is easily shown to be FPT (Fixed-Parameter Tractable). We just take $M$ and compute all the $d!$ permutations of its critical characters (while keeping the order of the $d$ colored blocks). Then we have $d!$ sequences $M_1,M_2,...,M_{d!}$ and we just need to check if one of them appear in $A$ as a subsequence in both colors and critical characters.

\begin{corollary}
Parameterized by the size of $M$, OMDCI+ is FPT.
\end{corollary}

In the next section, we consider
the 0-OMDCI problem: what is the difficulty to decide that there is no malware in an input program $A$ at all (even when a malware pattern $M$ is given)? It turns out that this problem is co-NP-hard and the proof is more involved.

\section{0-OMDCI is co-NP-hard}

We state the main result in this section (also in this paper) as follows. Recall that 0-OMDCI is the problem of deciding if OMDCI has an optimal solution of size zero.
\begin{theorem}\label{thm:0-OMDCI}
Deciding if 0-OMDCI has a solution is co-NP-hard.
\end{theorem}


We show Theorem~\ref{thm:0-OMDCI} by producing a reduction from the complement Hamiltonian-Cycle problem (co-Hamiltonian Cycle problem) to 0-OMDCI.
Before we prove the theorem, we describe the reduction instance in detail as well as show some properties of the reduction instance.
We start by fixing an instance of the co-Hamiltonian Cycle problem, $G=(V,E)$ and choosing arbitrary ordering of the vertices $<v_1,v_2,v_3,\ldots,v_n>$.

\subsection{Reduction Instance}\label{subsec:instance}
We start by describing the set of characters \(\Sigma\) followed by the set of colors for the indices \(\Gamma\). For $(v_i,v_j)\in E$, we also write $v_i \sim_E v_j$.

For every $i \in [1\ldots n]$, we define the following set of characters:
\[ \Sigma^{(P)}_i := \{p_{i,j}: j \in [1 \ldots n]\}, \Sigma^{(S)}_i := \{s_{i,j}: j \in [1 \ldots n]\}, \Sigma^{(T)}_i := \{t_{i,j}: j \in [1 \ldots n]\};\]
and
\[ \Sigma_i = \Sigma^{(P)}_i \cup \Sigma^{(S)}_i \cup \Sigma^{(T)}_i.\]
We then define our set of characters as \(\Sigma = \cup_{i}^{n} \Sigma_{i}\).

Similarly to the sets of characters, we define the following sets of colors.
\[ \Gamma^{(P)}_i := \{P_{i,j}: j \in [1 \ldots n]\}, \Gamma^{(S)}_i := \{S_{i,j}: j \in [1 \ldots n]\},
\Gamma^{(T)}_i := \{T_{i,j}: j \in [1 \ldots n]\},\]
and
\[ \Gamma_i = \Gamma^{(P)}_i \cup \Gamma^{(S)}_i \cup \Gamma^{(T)}_i.\]
Like before, define our set of colors as \(\Gamma = \cup_{i}^{n} \Gamma _{i}\).

The input to the 0-OMDCI instance \(M\) and \(A\) are constructed using two types of gadgets named
\(Selection\) and \(Linking\).
The \(linking\) gadget further has two subgadgets called \(PreLink\) and \(PostLink\).
Every gadget consists of a string with colored indices for both \(A\) and \(M\).
Before moving on to the description of the gadgets, we present the intuition used in their construction and the purpose of the different sets used to construct \(\Sigma\) and \(\Gamma\).

\begin{idea}
    Elements of the sets \(\Sigma^{(P)}_i\) and \(\Gamma^{(P)}_i\) are for \textbf{picking} vertices.
    The colors in \(\Gamma^{(P)}_i\) appear only in the \(i\)th \(Selection\) gadget.
    \label{idea:pickvertices}
\end{idea}

\begin{idea}
    The elements of the sets \(\Sigma^{(P)}_{i}\) and \(\Sigma^{(P)}_{1 + (i \% n )}\) are connected together by the set of colors \(\Gamma^{(T)}_i\) which serves as a mechanism for \textbf{linking} adjacent vertices in the Hamiltonian cycle. Furthermore, through the colors in \(\Gamma^{(T)}_i\) they connect the Linking subgadgets in $M$ and $A$ with selection subgadgets in $M$ and $A$.
\end{idea}

\begin{idea}
    Elements of sets \(\Sigma^{(S)}_i\) and \(\Gamma^{(S)}_i\) are for \textbf{starting} an interval in \(A\).
    Elements of sets \(\Sigma^{(T)}_i\) and \(\Gamma^{(T)}_i\) are for \textbf{terminating} an interval in \(A\). The intervals prevents selecting a vertex more than once.
    The colors in \(\Gamma^{(S)}_i\) appear only in \(PreLink\) subgadgets of \(Link\) gadgets and the colors in \(\Gamma^{(T)}_i\) only appear in \(PostLink\) subgadgets.
\end{idea}

Now define the \(Selection\) gadget which has \(n\)
components.
For every \(i \in [n]\), we define the following strings with colored intervals, with colors appearing at the bottom:
\begin{equation}\label{eq:mselwi}
     MSelection_i := \frac{p_{i,1}}{P_{i,1}}\cdot \frac{p_{i,2}}{P_{i,2}}\cdot
                            \ldots \cdot\frac{p_{i,n}}{P_{i,n}},
\end{equation}

\begin{equation}\label{eq:aselwi}
     ASelection_i := \frac{s_{i,n}}{P_{i,n}}\cdot\frac{s_{i,n-1}}{P_{i,n-1}}\cdot
                            \ldots \cdot\frac{s_{i,1}}{P_{i,1}}.
\end{equation}
The string with colored indices for \(M\) and \(A\) are
\begin{equation}\label{eq:msel}
     MSelection := \prod_{i \in [n]} MSelection_i
\end{equation}
and
\begin{equation}\label{eq:asel}
     ASelection := \prod_{i \in [n]} ASelection_i
\end{equation}
respectively.

There is a \(PreLink\) and \(PostLink\) subgadget for every two
indices $i, j \in [n]$ and a \(Linking\) subgadget for index \(j \in [n]\).
We start by specifying only the \(Prelink\) subgadget.

For $j \in [1\ldots n]$, let $\alpha[j,1], \ldots ,\alpha[j,deg(v_j)]$ be the list of all indices for vertices adjacent to $v_j$ appearing increasing order.
Fixing \(j\), for \(i \in [1..n]\) we define the strings with colored indices for \(M\) in the \(PreLink\) subgadget as
\begin{equation}\label{eq:mprelink-with-ij-upto-n-1}
     MPreLink_{i,j} = \frac{s_{i,j}}{S_{i,j}}.
\end{equation}
For the string with colored indices for \(A\) in the \(PreLink\) subgadget, we similarly define for \(i \in [n]\):
\begin{equation}\label{eq:aprelink-with-ij-upto-n-1}
      APreLink_{i,j} = \frac{t_{i,j}}{S_{i,j}}.
\end{equation}

Proceeding to the \(PostLink\) subgadget,
For all \(i,j \in [n] \), we define its subgadgets for the \(M\) string as
\begin{equation}\label{eq:mpostlink-with-ij}
      MPostLink_{i,j} =  \frac{t_{i,j}}{T_{i,j}}.
\end{equation}
For any \(j \in [n]\), we define the \(PostLink\) subgadgets for the \(A\)
string as follows:

\begin{equation}\label{eq:apostlink-with-ij}
    APostLink_{i,j} = \frac{p_{1 + (i\% n), \alpha[j,1]}}{T_{i,j}} \cdot\frac{p_{1 + (i\% n), \alpha[j,2]}}{T_{i,j}} \cdots \frac{p_{1 + (i\% n), \alpha[j,deg(v_j)]}}{T_{i,j}}.
\end{equation}


To complete the construction of \(Linking\) gadgets for every \(j \in [n]\), we use \(PreLink_{i,j}\) and \(PostLink_{i,j}\) subgadgets where \(i\in[n]\) as follows: 
\begin{equation}\label{eq:mlinking-with-j}
       MLinking_j =  \prod_{i \in [n]}( MPreLink_{i,j} \circ MPostLink_{i,j})
\end{equation}
and
\begin{equation}\label{eq:alinking-with-j}
      ALinking_j = \left( \prod_{i \in [n]} APreLink_{i,j} \right)
    \circ   \left(\prod_{i \in [n]} APostLink_{i,j}\right)
\end{equation}
for the \(M\) and \(A\) strings respectively.
We construction the complete string with colored for indices for the \(M\) and \(A\) in gadget as:
\begin{equation}\label{eq:mlinking}
      MLinking = \prod_{j \in [n]} MLinking_{j},
\end{equation}
\begin{equation}\label{eq:alinking}
      ALinking = \prod_{j \in [n]} ALinking_{j}.
\end{equation}

The complete reduction instance for \(G\) now may be computed by concatenating
our two gadgets together with
\begin{equation}\label{eq:m-construction}
    M = MSelection \circ MLinking
\end{equation}
and
\begin{equation}\label{eq:a-construction}
    A = ASelection \circ ALinking.
\end{equation}
A flow of the reduction can be found in Figure~\ref{fig1}.
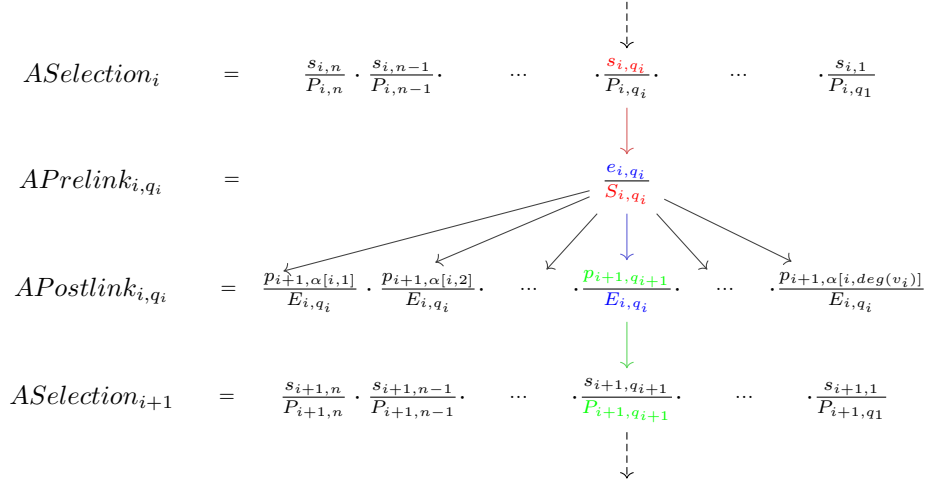
\begin{figure}
    \centering
    \[\begin{tikzcd}
	&& {} \\
	{ASelection_i} & {\frac{s_{i,n}}{P_{i,n}}\cdot \frac{s_{i,n-1}}{P_{i,n-1}}\cdot} & {\cdot\frac{\color{red}s_{i,q_{i}}}{P_{i,q_{i}}}\cdot} & {\cdot\frac{s_{i,1}}{P_{i,q_1}}} \\
	{APrelink_{i,q_i}} & {} & {\frac{\color{blue}e_{i,q_i}}{\color{red}S_{i,q_i}}} & {} \\
	{APostlink_{i,q_i}} & {\frac{p_{i+1,\alpha[i,1]}}{E_{i,q_i}}\cdot\frac{p_{i+1,\alpha[i,2]}}{E_{i,q_i}}\cdot} & {\cdot\frac{\color{green}p_{i+1,q_{i+1}}}{\color{blue}E_{i,q_i}}\cdot} & {\cdot\frac{p_{i+1,\alpha[i,deg(v_i)]}}{E_{i,q_i}}} \\
	{ASelection_{i+1}} & {\frac{s_{i+1,n}}{P_{i+1,n}}\cdot \frac{s_{i+1,n-1}}{P_{i+1,n-1}}\cdot} & {\cdot\frac{s_{i+1,q_{i+1}}}{\color{green}P_{i+1,q_{i+1}}}\cdot} & {\cdot\frac{s_{i+1,1}}{P_{i+1,q_1}}} \\
	&& {}
	\arrow[dashed, from=1-3, to=2-3]
	\arrow["{=}"{description}, draw=none, from=2-1, to=2-2]
	\arrow["\ldots"{description}, draw=none, from=2-2, to=2-3]
	\arrow["\ldots"{description}, draw=none, from=2-3, to=2-4]
	\arrow[color={rgb,255:red,214;green,92;blue,92}, from=2-3, to=3-3]
	\arrow["{=}"{description, pos=0.3}, draw=none, from=3-1, to=3-2]
	\arrow[color={rgb,255:red,64;green,64;blue,64},  from=3-3, to=4-2]
	\arrow[color={rgb,255:red,92;green,92;blue,214}, from=3-3, to=4-3]
	\arrow[color={rgb,255:red,64;green,64;blue,64},  from=3-3, to=4-4]
	\arrow[""{name=0, anchor=center, inner sep=0}, "{=}"{description, pos=0.7}, draw=none, from=4-1, to=4-2]
	\arrow[""{name=1, anchor=center, inner sep=0}, "\ldots"{description}, draw=none, from=4-2, to=4-3]
	\arrow[""{name=2, anchor=center, inner sep=0}, "\ldots"{description}, draw=none, from=4-3, to=4-4]
	\arrow[color={rgb,255:red,92;green,214;blue,92}, from=4-3, to=5-3]
	\arrow["{=}"{description}, draw=none, from=5-1, to=5-2]
	\arrow["\ldots"{description}, draw=none, from=5-2, to=5-3]
	\arrow["\ldots"{description}, draw=none, from=5-3, to=5-4]
	\arrow[dashed, from=5-3, to=6-3]
	\arrow[color={rgb,255:red,64;green,64;blue,64}, shorten >=26pt, from=3-3, to=0]
	\arrow[color={rgb,255:red,64;green,64;blue,64}, shorten >=9pt, from=3-3, to=1]
	\arrow[color={rgb,255:red,64;green,64;blue,64}, shorten >=9pt, from=3-3, to=2]
\end{tikzcd}\]
    \caption{An illustration of the flow of the reduction.}
    \label{fig1}
\end{figure}

\subsection{Properties of reduction instance}
In this section, we show properties of $M$ and $A$ along with the properties of any OMDCI solution between $M$ and $A$. Let \(\beta:\Gamma \cup \Sigma' \rightarrow [1...n]\) map each color and character to the second index of their subscript; e.g., \(\beta(p_{i,j}) = j\) and \(\beta(P_{i+1,\alpha[j,1]}) = \alpha[j,1]\). 


\subsubsection{Properties of $A$ and $M$}
\begin{observation}
\(M\) is a permutation string and each color and critical character in $M$ appears exactly once.
\label{obs:eachcolorappearsonceinM}
\end{observation}
\begin{observation}
    If \( c \in \Gamma \setminus \cup_{i \in [n]} \Gamma^{(T)}_i\), then exactly one index is colored \(c\) in \(A\).
    \label{obs:exactlyonecolorexceptedges}
\end{observation}
\begin{observation}
    If \(c \in \Gamma^{(T)}_i\) and \(\beta(c) = i\)
    then the number of indices in \(A\) 
    colored \(c\) is \(\deg(v_i)\)
    and all the indices colored \(c\) appear in one \(APostLink_{i,j}\) subgadget.
    \label{obs:endsoccurdegeachunique}
\end{observation}

\begin{observation}
    If \(\delta \in \Sigma \setminus \cup_{i \in [n]}\Sigma_i^{(P)}\), then \(\delta\) appears exactly once in \(A\).
    \label{obs:exactlyoneciexceptpicksci}
\end{observation}

\subsubsection{Properties of OMDCI solutions between $A$ and $M$}
Suppose there exists a solution of size \( k > 0 \) for the input pair \((A, M)\). Using the same notation as in Definition \ref{def:gomdic} and the preliminaries, let \( M' \) and \( A' \) denote the common subsequences of colors. Let \( I_{M'} = \{a_1, a_2, \ldots, a_k\} \) be the strictly increasing indices in \(\gamma_M\) corresponding to \( M' \), and \( I_{A'} = \{b_1, b_2, \ldots, b_k\} \) be the strictly increasing indices in \(\gamma_A\) corresponding to \( A' \). 

\begin{proposition}
For any \( l \in [k] \), the color \( \gamma_M(a_l) \) appears at some index in \( MSelection_i \) if and only if the color \( \gamma_A(b_l) \) appears at some index in \( ASelection_i \).
\label{prop:mselectmatchaselect}
\end{proposition}

\begin{proof}
This proposition essentially states that if a color is picked in some \(MSelection_i\), then its corresponding colored index must only exist in \(ASelection_i\) and vice versa.

Immediately following Observations \ref{obs:eachcolorappearsonceinM} and \ref{obs:exactlyonecolorexceptedges}, any color that is matched in either $MSelection_i$ or $ASelection_i$ must be matched solely within their respective gadgets.
\end{proof}
\begin{proposition}
    For any \(l \in [k]\), the color \(\gamma_M(a_l)\) appears at some index in \(MPreLink_{i,j}\) if and only if the color \(\gamma_A(b_l)\) appears at some index in \(APreLink_{i,j}\).
    \label{prop:mprelink-matches-aprelink}
\end{proposition}

\begin{proof}
    Similar to Proposition \ref{prop:mselectmatchaselect}, this proposition states that if a color is picked in some \(MPreLink_{i,j}\), then its corresponding colored index must only exist in \(APreLink_{i,j}\), and vice versa.
    
    The proof is straightforward. From Observations \ref{obs:eachcolorappearsonceinM} and \ref{obs:exactlyonecolorexceptedges}, any color matched in either $MPreLink_{i,j}$ or $APreLink_{i,j}$ must be matched solely within their respective gadgets.
\end{proof}

\begin{proposition}
    For any \(l \in [k]\), the color \(\gamma_M(a_l)\) appears at some index in \(MPostLink_{i,j}\) if and only if the color \(\gamma_A(b_l)\) appears at some index in \(APostLink_{i,j}\).
    \label{prop:mpostlink-matches-mpostlink}
\end{proposition}
\begin{proof}
    The argument is similar to the previous two proofs. Immediately following Observations \ref{obs:eachcolorappearsonceinM} and \ref{obs:endsoccurdegeachunique}, any color that is matched in either $MPostLink_{i,j}$ or $APostLink_{i,j}$ must be matched solely within their respective gadgets.
\end{proof}

\subsubsection{Properties from $Mselection_i$ and $ASelection_i$}

\begin{proposition}\label{prop:one-per-selection}
For all $i \in [n]$, any OMDCI solution can only include at most one colored index from $MSelection_i$ and $ASelection_i$ selection subgadgets.
\end{proposition}
\begin{proof}
Directly following the Proposition \ref{prop:mselectmatchaselect} and the fact that colors of $ASelection_i$ is the reversed sequence of $MSelection_i$ colors, the longest common subsequence of colors between the said subgadgets is of size $1$. Consequently, at most one colored index from $MSelection_i$ and $ASelection_i$ can be included in any OMDCI solution.
\end{proof}
\begin{proposition}\label{prop:p-to-s}
For all $i \in [n]$, If the colored interval $P_{i,j}$ is included in the OMDCI solution, then the colored interval $S_{i,j}$ is included in the solution.
\end{proposition}

\begin{proof}
First observe that for any \(i \in [m]\) the critical characters $p_{i,j}$ and $s_{i,j}$ appear exactly once in $M$ and are associated with the colored intervals $P_{i,j}$ and $S_{i,j}$, respectively, in the selection gadgets $MSelection_i$ and $MPreLink_{i,j}$. Now, assume that the color $P_{i,j}$ is included in the OMDCI solution; this assumption coupled with Proposition \ref{prop:mselectmatchaselect} implies that the solution contains the critical characters $p_{i,j}$ and $s_{i,j}$. Following the previously mentioned observation and the Property \ref{property:setofcoveredcriticalinstructionsareequal} of the OMDCI solution; it is immediately evident that the solution must contain the color $S_{i,j}$ to cover the critical character $s_{i,j}$. 
\end{proof}

\subsubsection{Properties from the \(Linking\) Gadgets}\label{subsec:linking-prop}
We now show some properties dependent on the construction of the \(Linking\)
gadget.
The first group of properties show how the \(Linking\) gadget ``links'' 
colors in \(\Gamma\) together. 

\begin{proposition}\label{prop:s-to-e}
    For any \(i,j \in [n]\), 
    if there exist some index in \(M'\) and \(A'\) with color \(S_{i,j}\)
    then there must exist some index in \(M'\) and \(A'\) with
    color \(T_{i,j}\).
\end{proposition}
\begin{proof}
    The proof follows directly from Observations \ref{obs:eachcolorappearsonceinM}, and \ref{obs:exactlyonecolorexceptedges}. Assume that the color $S_{i,j} \in MPreLink_{i,j}$ is selected in the solution. This assumption, combined with Proposition \ref{prop:mprelink-matches-aprelink}, implies that the critical characters $s_{i,j}$ and $t_{i,j}$ are selected in $M'$ and $A'$, respectively. Consequently, the solution must also select the color $T_{i,j}$, as it is the only color associated with the critical character $t_{i,j}$, in order to satisfy Property \ref{property:setofcoveredcriticalinstructionsareequal}.
\end{proof}

\begin{proposition}\label{prop:e-to-l-p}
    For any \(i,j \in [n]\), 
    when there exist some index in \(M'\) and \(A'\) with color \(T_{i,j}\)
    then the following hold:
there must exist some index in \(M'\) and \(A'\) colored \( P_{1 + (i \% n), k} \), where \(v_k \sim_E v_j\).
\end{proposition}
\begin{proof}
First let us assume that the solution contains the color \(T_{i,j}\) for some \(i,j \in [n]\). Note that the color \(T_{i,j}\) appears only in the gadget \(APostLink_{i,j}\) within \(A\). Consequently, the presence of the color \(T_{i,j}\) in \(A'\), specifically in \(APostLink_{i,j}\), must include a critical character in the form \(p_{1 + (i\%n), \alpha[j,q]}\) where \(q \in [deg(v_j)]\) and \(k = \alpha[j,q]\). Since \(p_{1 + (i\%n), k}\) is part of the solution, it implies, by Observation \ref{obs:eachcolorappearsonceinM} as well as Equation \ref{eq:mselwi}, the solution must also contain the color \(P_{1 + (i\%n), k}\). Furthermore, according to Equation \ref{eq:apostlink-with-ij}, \(p_{1 + (i\%n), k}\) exists in \(APostLink_{i,j}\) only if there exists \(v_j \sim_E v_k\). This proves the proposition.
\end{proof}

Combining Propositions \ref{prop:p-to-s}, \ref{prop:s-to-e}, and \ref{prop:e-to-l-p},
 yields the following:
\begin{lemma}\label{prop:p-to-p}
    For any \(i,j \in [n]\), 
    if there exist some index in \(M'\) and \(A'\) with color \(P_{i,j}\)
    then there must exist some index \(k \in [n]\) such that there exist \(v_j \sim v_k\) and the colors \(P_{1 + (i \% n), k}\) and \(S_{1 + (i \% n), k}\) exist in \(M'\) and \(A'\).
\end{lemma}
\begin{proof}
First, assume the premise. Then, by sequentially applying Modus Ponens to Propositions \ref{prop:p-to-s}, \ref{prop:s-to-e}, and \ref{prop:e-to-l-p}, we conclude that the color \(P_{1 + (i \% n), k}\) exists in both \(M'\) and \(A'\). Applying the Proposition \ref{prop:p-to-s} once more, we further deduce that the color \(S_{1 + (i \% n), k}\) must also exist in \(M'\) and \(A'\). This completes the proof.
\end{proof}

\begin{figure}[t]
\centering
\includegraphics[width=0.50\textwidth]{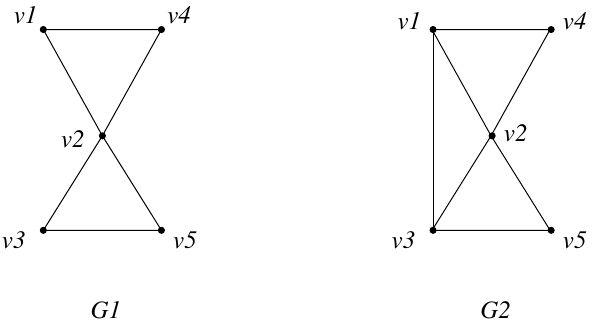}
\caption{\label{fig:fig1}The reductions from two graphs $G_1$ and $G_2$ would create different OMDCI instances.}
\label{fig2}
\end{figure}

We use $G_2$ in Figure~\ref{fig2} as an example to illustrate this lemma. In this case, consider the Hamiltonian cycle: $\langle 5,3,1,4,2\rangle$. We first match the color $P_{1,5}$ in \(MSelection_1\) and \(ASelection_1\). Then \(MPreLink_{1,5} \) enforces the selection of $S_{1,5}$, which is selected in turn by \(APreLink_{1,5}\), together with the critical character $t_{1,5}$. Finally, \(MPostLink_{1,5}\)  would need to match the character $t_{1,5}$, which would also pick the color $T_{1,5}$. Consequently, \(APostLink_{1,5}\) would select
a critical character at step 2, which is $p_{2,3}$. Then we would repeat the above process by selecting the color $P_{2,3}$ and then $S_{2,3}$. We show the following flows (not the exact sequences, as $p_{i,j}$'s would be put together at the beginning of $M$ and $s_{i,j}$'s would be put together at the beginning of $A$) of $M'$ and $A'$, only giving details for the first selected vertex (among the 5 selected vertices):

 \[Flow~of~M'=MSelection_1\rightarrow MPrelink_{1,5}\rightarrow MPostLink_{1,5}\rightarrow MSelection_2 \cdots,\]
which expands to
\[ Flow~of~M' = \frac{p_{1,5}}{P_{1,5}}\rightarrow \frac{s_{1,5}}{S_{1,5}}\rightarrow \frac{t_{1,5}}{T_{1,5}}
\rightarrow\frac{p_{2,3}}{P_{2,3}}\cdots \frac{p_{3,1}}{P_{3,1}} \cdots \frac{p_{4,4}}{P_{4,4}} \cdots \frac{p_{5,2}}{P_{5,2}} \cdots\frac{t_{5,2}}{T_{5,2}} ;\]
and,
 \[Flow~of~A'=ASelection_1\cdot APrelink_{1,5}\cdots APostLink_{1,5}\circ ASelection_2 \cdots,\]
which expands to
\[ Flow~of~A' = \frac{s_{1,5}}{P_{1,5}}\rightarrow \frac{t_{1,5}}{S_{1,5}}\rightarrow \frac{p_{2,3}}{T_{1,5}}
\rightarrow\frac{s_{2,3}}{P_{2,3}}\rightarrow \frac{s_{3,1}}{P_{3,1}} \cdots \frac{s_{4,4}}{P_{4,4}} \cdots \frac{s_{5,2}}{P_{5,2}} \cdots\frac{p_{1,5}}{T_{5,2}} .\]
Note that the first character $p_{1,5}$ in $M'$ matches the last character $p_{1,5}$ in $A'$.

The following proposition is useful in showing the \(Linking\) gadget prevents
vertices from being selected multiple times. 
\begin{proposition}\label{prop:no_l_repeats}
    For any \(i,t \in [n]\) where \(i \neq t\), if there exists \(k \in [n]\) such that \(S_{t, k}\) is included in \(M'\) and \(A'\), then  \(S_{i,k}\) can not be included in the solution.
\end{proposition}
\begin{proof}

    Assume, without loss of generality, that \(t < i\). We proceed by negating the claim. Specifically, suppose that for some \(i, k, t \in [n]\) and the colors \(S_{t,k}\) and \(S_{i,k}\) are included in \(M'\) and \(A'\).  
    Note that $S_{t,k}$ and $S_{i,k}$ appear in the $MLinking_k$ and $ALinking_k$ gadgets within $M$ and $A$, respectively. More specifically, $S_{t,k}$ appears in $MPreLink_{t,k}$ and $APreLink_{t,k}$, while $S_{i,k}$ appears in $MPreLink_{i,k}$ and $APreLink_{i,k}$.  

    By Proposition \ref{prop:s-to-e}, we also deduce that the colors $T_{t,k}$ and $T_{i,k}$ must be included in $M'$ and $A'$. These colors reside in the $MLinking_k$ and $ALinking_k$ gadgets, where $T_{t,k}$ appears in $MPostLink_{t,k}$ and $APostLink_{t,k}$, and $T_{i,k}$ appears in $MPostLink_{i,k}$ and $APostLink_{i,k}$.  

    Thus, the colors $S_{t,k}, S_{i,k}, T_{t,k}, T_{i,k}$ must all appear in $M'$ and $A'$. Now, consider the order of these subgadgets within $MLinking_k$ and $ALinking_k$, as described in Equations \ref{eq:alinking-with-j} and \ref{eq:mlinking-with-j}:  
    
    \begin{itemize}
        \item In $MLinking_k$, for all $j \in [n]$, the sequence $MPreLink_{j,k} \circ MPostLink_{j,k}$ appears consecutively. Hence, $MPreLink_{t,k} \circ MPostLink_{t,k}$ must appear before $MPreLink_{i,k} \circ MPostLink_{i,k}$.  
        \item In $ALinking_k$, for all $j \in [n]$, the $APreLink_{j,k}$ subgadgets always appear before any $APostLink_{j,k}$ subgadgets.
    \end{itemize}
    By Propositions \ref{prop:mprelink-matches-aprelink} and \ref{prop:mpostlink-matches-mpostlink}, the color $S_{t,k}$ is matched within $MPreLink_{t,k}$  and $APreLink_{t,k}$, and the color $T_{t,k}$ is matched within $MPostLink_{t,k}$ and $APostLink_{t,k}$. Similarly, by Proposition \ref{prop:mprelink-matches-aprelink}, the color $S_{i,k}$ must be matched within $MPreLink_{i,k} \text{ and } APreLink_{i,k}$. 
    However, this order of colors is impossible to achieve in the common subsequence, leading to a contradiction. Thus, the initial assumption must be false.  
\end{proof}

We use $G_1$ in Figure~\ref{fig2} as an example to illustrate this observation. In this case, consider the path we attempt to include in the Hamiltonian cycle:
$\langle 1,2,3,5,2\rangle$. 
Initially, the color match $P_{1,1}$ is matched within $MSelection_1$ and $ASelection_1$.
Then \(MPreLink_{1,1} \) enforces the selection of $S_{1,1}$, which is selected in turn by \(APreLink_{1,1}\), together with the critical character $t_{1,1}$. Finally, \(MPostLink_{1,1}\)  would need to match the character $t_{1,1}$, which would also pick the color $T_{1,1}$. Consequently, \(APostLink_{1,1}\) would select
a critical character starting the selection of the second vertex in the cycle, which is $p_{2,2}$. Then we would repeat the above process by selecting the color $P_{2,2}$ and then $S_{2,2}$, etc. We show the following flows (again, not the exact sequences) of $M'$ and $A'$, only giving details for the first selected vertex (among the 5 selected vertices):
\begin{equation}
    \begin{split}
        &Flow~of~M'=\\
        &MSelection_1\rightarrow MPrelink_{1,5}\rightarrow MPostLink_{1,5}\rightarrow MSelection_2 \cdots,
    \end{split}
\end{equation}
which expands to
\[ Flow~of~M' = \frac{p_{1,1}}{P_{1,1}}\rightarrow \frac{s_{1,1}}{S_{1,1}}\rightarrow \frac{t_{1,1}}{T_{1,1}}
\rightarrow\frac{p_{2,2}}{P_{2,2}}\rightarrow\frac{s_{2,2}}{S_{2,2}}\cdots \frac{s_{3,3}}{S_{3,3}} \cdots \frac{s_{4,5}}{S_{4,5}} \rightarrow\frac{t_{4,5}}{T_{4,5}},\]
and
 \[Flow~of~A'=ASelection_1\rightarrow APrelink_{1,5}\rightarrow APostLink_{1,5}\rightarrow ASelection_2 \cdots,\]
which expands to
\[ Flow~of~A' = \frac{s_{1,1}}{P_{1,1}}\rightarrow \frac{t_{1,1}}{S_{1,1}}\rightarrow \frac{p_{2,2}}{T_{1,1}}
\rightarrow\frac{s_{2,2}}{P_{2,2}}\rightarrow \frac{t_{2,2}}{S_{2,2}}\cdots \frac{t_{3,3}}{S_{3,3}} \cdots \frac{t_{4,5}}{S_{4,5}} \rightarrow\frac{p_{5,2}}{T_{4,5}} .\]

Note that $p_{5,2}$ never appears in $M'$. In fact, since $G_1$ has no Hamiltonian cycle, even if we add the character $p_{5,2}$ to the end of $A'$ it can never match the first character $p_{1,1}$ in $M'$.

\subsection{Proof of Theorem~\ref{thm:0-OMDCI}}

We now proceed to prove Theorem~\ref{thm:0-OMDCI}, as the ``if and only if'' relation is a bit long, we separate it out as the following lemma.


\begin{lemma}\label{lemma:omdci-0-iff}
The graph \(G\) has a Hamiltonian cycle if and only if there exists \(k > 0\) such that there exist subsequences \(M'\) of \(\gamma_M\) and \(A'\) of \(\gamma_A\) with \(I_{M'} = \{a_1, a_2, \ldots , a_k\}\), \(I_{A'} =\{b_1, b_2, \ldots , b_k\}\) such that properties of definition \ref{def:gomdic} holds.

\end{lemma}
Due to space constraint, we leave the proof 
of Lemma~\ref{lemma:omdci-0-iff} to the appendix.
We now complete the proof of Theorem~\ref{thm:0-OMDCI}.
 
\begin{proof}
    The construction of the \(Selection\)
    subgadgets uses \(O(n)\) characters thus
    the whole \(Selection\) gadget consists of
    \(O(n^2)\) characters.
    The \(PreLink_{i,j}\) and \(PostLink_{i,j}\) subgadgets are of size
    1 in \(M\) and of size \(1\) and \(deg(v_j)\)
    in \(A\). Thus, \(MLinking_{j}\) and \(ALinking_{j}\) are of sizes \(2n\) and
    \(n(1+deg(v_j))\) respectively.
    Consequently the size of \(MSelection\) gadget
        is \(2n^2\) and the size of \(ASelection\) is  \(\sum_j^{n} n(1+deg(v_j)) = n^2 + 2n|E|=O(n^3)\).
    Therefore the construction of \(M\) and \(A\) can be
    done in \(O(n^2 + n^3) = O(n^3)\) time.
    As \(G = (V,E)\) has no Hamiltonian cycle
    if and only if \(M\) and \(A\) have an optimal solution of size zero by Lemma~\ref{lemma:omdci-0-iff}, we have deciding 
    \omdci has a zero solution is co-NP-hard.
\end{proof}

Suppose there is a polynomial-time approximation algorithm ${\cal A}$ for OMDCI, then running ${\cal A}$ on the above instance, it would solve the instance (i.e., deciding that there is a size-zero solution) in polynomial time. We then have the following corollary.

\begin{corollary}
The maximization version of OMDCI has no polynomial-time approximation (regardless of its factor) unless P=co-NP.
\end{corollary}

Similarly, we have the following corollary regarding the fixed-parameter tractability of OMDCI.
\begin{corollary}
OMDCI has no FPT algorithm parameterized by the solution size unless P=co-NP.
\end{corollary}

\begin{proof}
Suppose there is an FPT algorithm ${\cal F}$ for OMDCI parameterized by the solution size $d$, which runs in $O(f(d)n^c)$ time, where $f(-)$ is any computable function and $c$ is a constant not related to the input length $n$. At the end of the instance $M$ and $A$ in Theorem~\ref{thm:0-OMDCI}, attach $2$ unique new blocks each with a distinct new letter. Running ${\cal F}$ on this new instance, if ${\cal F}$ returns a size-$2$ solution, then it would imply that the original instance has a zero solution.
By Theorem~\ref{thm:0-OMDCI}, the latter would further imply that ${\cal F}$ solves 
co-HC in polynomial (i.e., $O(f(2)n^c)$) time.
This is not possible unless P=co-NP. Therefore, OMDCI has no FPT algorithm parameterized by the solution size $d$ unless P=co-NP.
\end{proof}

A twist of the above proof, by attaching some intervals at the end of $M$ and $A$ such that they must all appear in the optimal solution, also implies that OMDCI is co-NP-hard. Hence we additionally have the following corollary.

\begin{corollary}
OMDCI is NP-complete and co-NP-hard.
\end{corollary}

\section{Concluding Remarks}

In this paper, we investigated the computational complexity of detecting malware in low-level programs with feature matching algorithms. Essentially, we show that the problem is hard in both directions---determining the presence of some malware and verifying its absence are both computationally hard. Our method more generally shows that matching features in order data is computationally difficult. For example, trying to find a common sequence of events between two calendars which share the same participants forms a \omdci instance where critical characters are participants and the colors are the types of events.


Certainly---in reality---the scenario could be more complex. For instance, the new malware might be embedded in a program written in architecture B while the known malware is written in architecture C where B is based on a 32-bit machine while C is on a 64-bit machine. In this case, the set of instructions in B and C are different, even the number of instructions may not be the same, in which case one cannot obtain a one-to-one mapping between the instructions. In this case, it should be harder to
compare the similarity of two programs written in B and C directly. One idea is to use additional semantic information or structure of such programs; for instance, the known function of some blocks, or even the calling graph of blocks. We are currently working along this line.



\begin{thebibliography}{10}

\bibitem{Aonzo23}
Simone Aonzo, Yufei Han, Alessandro Mantovani, and Davide Balzarotti.
\newblock Humans vs. machines in malware classification.
\newblock In Joseph~A. Calandrino and Carmela Troncoso, editors, {\em 32nd {USENIX} Security Symposium, {USENIX} Security 2023, Anaheim, CA, USA, August 9-11, 2023}, pages 1145--1162. {USENIX} Association, 2023.

\bibitem{Bonett18}
Richard Bonett, Kaushal Kafle, Kevin Moran, Adwait Nadkarni, and Denys Poshyvanyk.
\newblock Discovering flaws in security-focused static analysis tools for android using systematic mutation.
\newblock In William Enck and Adrienne~Porter Felt, editors, {\em 27th {USENIX} Security Symposium, {USENIX} Security 2018, Baltimore, MD, USA, August 15-17, 2018}, pages 1263--1280. {USENIX} Association, 2018.

\bibitem{DBLP:conf/uss/ChristodorescuJ03}
Mihai Christodorescu and Somesh Jha.
\newblock Static analysis of executables to detect malicious patterns.
\newblock In {\em Proceedings of the 12th {USENIX} Security Symposium, Washington, D.C., USA, August 4-8, 2003}. {USENIX} Association, 2003.

\bibitem{Cook71}
Stephen~A. Cook.
\newblock The complexity of theorem-proving procedures.
\newblock In Michael~A. Harrison, Ranan~B. Banerji, and Jeffrey~D. Ullman, editors, {\em Proceedings of the 3rd Annual {ACM} Symposium on Theory of Computing, May 3-5, 1971, Shaker Heights, Ohio, {USA}}, pages 151--158. {ACM}, 1971.

\bibitem{Coscia23}
Antonio Coscia, Vincenzo Dentamaro, Stefano Galantucci, Antonio Maci, and Giuseppe Pirlo.
\newblock {YAMME:} a yara-byte-signatures metamorphic mutation engine.
\newblock {\em {IEEE} Trans. Inf. Forensics Secur.}, 18:4530--4545, 2023.

\bibitem{DF99}
Rodney~G. Downey and Michael~R. Fellows.
\newblock {\em Parameterized Complexity}.
\newblock Monographs in Computer Science. Springer, 1999.

\bibitem{FG06}
J{\"{o}}rg Flum and Martin Grohe.
\newblock {\em Parameterized Complexity Theory}.
\newblock Texts in Theoretical Computer Science. An {EATCS} Series. Springer, 2006.

\bibitem{Gagnon17}
Fran{\c{c}}ois Gagnon and Fr{\'{e}}d{\'{e}}ric Massicotte.
\newblock Revisiting static analysis of android malware.
\newblock In Jos{\'{e}}~M. Fernandez and Mathias Payer, editors, {\em 10th {USENIX} Workshop on Cyber Security Experimentation and Test, {CSET} 2017, Vancouver, BC, Canada, August 14, 2017}. {USENIX} Association, 2017.

\bibitem{DBLP:books/fm/GareyJ79}
M.~R. Garey and David~S. Johnson.
\newblock {\em Computers and Intractability: {A} Guide to the Theory of NP-Completeness}.
\newblock W. H. Freeman, 1979.

\bibitem{Hu13}
Xin Hu, Kang~G. Shin, Sandeep Bhatkar, and Kent Griffin.
\newblock Mutantx-s: Scalable malware clustering based on static features.
\newblock In Andrew Birrell and Emin~G{\"{u}}n Sirer, editors, {\em Proceedings of the 2013 {USENIX} Annual Technical Conference, {USENIX} {ATC} 2013, San Jose, CA, USA, June 26-28, 2013}, pages 187--198. {USENIX} Association, 2013.

\bibitem{Jin20}
Xiang Jin, Xiaofei Xing, Haroon Elahi, Guojun Wang, and Hai Jiang.
\newblock A malware detection approach using malware images and autoencoders.
\newblock In {\em 17th {IEEE} International Conference on Mobile Ad Hoc and Sensor Systems, {MASS} 2020, Delhi, India, December 10-13, 2020}, pages 1--6. {IEEE}, 2020.

\bibitem{Jindal19}
Chani Jindal, Christopher Salls, Hojjat Aghakhani, Keith Long, Christopher Kruegel, and Giovanni Vigna.
\newblock Neurlux: dynamic malware analysis without feature engineering.
\newblock In David~M. Balenson, editor, {\em Proceedings of the 35th Annual Computer Security Applications Conference, {ACSAC} 2019, San Juan, PR, USA, December 09-13, 2019}, pages 444--455. {ACM}, 2019.

\bibitem{Karp72}
Richard~M. Karp.
\newblock Reducibility among combinatorial problems.
\newblock In Raymond~E. Miller and James~W. Thatcher, editors, {\em Proceedings of a symposium on the Complexity of Computer Computations, held March 20-22, 1972, at the {IBM} Thomas J. Watson Research Center, Yorktown Heights, New York, {USA}}, The {IBM} Research Symposia Series, pages 85--103. Plenum Press, New York, 1972.

\bibitem{LiJia23}
Shijia Li, Jiang Ming, Pengda Qiu, Qiyuan Chen, Lanqing Liu, Huaifeng Bao, Qiang Wang, and Chunfu Jia.
\newblock Packgenome: Automatically generating robust {YARA} rules for accurate malware packer detection.
\newblock In Weizhi Meng, Christian~Damsgaard Jensen, Cas Cremers, and Engin Kirda, editors, {\em Proceedings of the 2023 {ACM} {SIGSAC} Conference on Computer and Communications Security, {CCS} 2023, Copenhagen, Denmark, November 26-30, 2023}, pages 3078--3092. {ACM}, 2023.

\bibitem{Lucas24}
Keane Lucas, Weiran Lin, Lujo Bauer, Michael~K. Reiter, and Mahmood Sharif.
\newblock Training robust ml-based raw-binary malware detectors in hours, not months.
\newblock In Bo~Luo, Xiaojing Liao, Jun Xu, Engin Kirda, and David Lie, editors, {\em Proceedings of the 2024 on {ACM} {SIGSAC} Conference on Computer and Communications Security, {CCS} 2024, Salt Lake City, UT, USA, October 14-18, 2024}, pages 124--138. {ACM}, 2024.

\bibitem{Morgan2020}
Steve Morgan.
\newblock Cybercrime to cost the world \$10.5 trillion annually by 2025.
\newblock Technical report, Cybercrime Magazine (on-line), https://cybersecurityventures.com/hackerpocalypse-cybercrime-report-2016/, 2020.

\bibitem{Seo23}
HyungBin Seo and MyungKeun Yoon.
\newblock Generative intrusion detection and prevention on data stream.
\newblock In Joseph~A. Calandrino and Carmela Troncoso, editors, {\em 32nd {USENIX} Security Symposium, {USENIX} Security 2023, Anaheim, CA, USA, August 9-11, 2023}, pages 4319--4335. {USENIX} Association, 2023.

\bibitem{DBLP:journals/tissec/Ugarte-PedreroG18}
Xabier Ugarte{-}Pedrero, Mariano Graziano, and Davide Balzarotti.
\newblock A close look at a daily dataset of malware samples.
\newblock {\em {ACM} Trans. Priv. Secur.}, 22(1):6:1--6:30, 2019.

\bibitem{Vu23}
Duc{-}Ly Vu, Zachary Newman, and John~Speed Meyers.
\newblock Bad snakes: Understanding and improving python package index malware scanning.
\newblock In {\em 45th {IEEE/ACM} International Conference on Software Engineering, {ICSE} 2023, Melbourne, Australia, May 14-20, 2023}, pages 499--511. {IEEE}, 2023.

\bibitem{Wu23}
Yafei Wu, Cong Sun, Dongrui Zeng, Gang Tan, Siqi Ma, and Peicheng Wang.
\newblock Libscan: Towards more precise third-party library identification for android applications.
\newblock In Joseph~A. Calandrino and Carmela Troncoso, editors, {\em 32nd {USENIX} Security Symposium, {USENIX} Security 2023, Anaheim, CA, USA, August 9-11, 2023}, pages 3385--3402. {USENIX} Association, 2023.

\bibitem{Zhong24a}
Fangtian Zhong, Xiuzhen Cheng, Dongxiao Yu, Bei Gong, Shuaiwen Song, and Jiguo Yu.
\newblock Malfox: Camouflaged adversarial malware example generation based on conv-gans against black-box detectors.
\newblock {\em {IEEE} Trans. Computers}, 73(4):980--993, 2024.

\bibitem{Zhong24}
Fangtian Zhong, Qin Hu, Yili Jiang, Jiaqi Huang, Cheng Zhang, and Dinghao Wu.
\newblock Enhancing malware classification via self-similarity techniques.
\newblock {\em {IEEE} Trans. Inf. Forensics Secur.}, 19:7232--7244, 2024.

\end{thebibliography}

\section{Appendix: proof of Lemma~24}
\begin{proof}
We start by showing the forward direction.
Assume \(G\) has a Hamiltonian cycle.
Let \(H\) be a Hamiltonian cycle in \(G\) and
assign the vertices in \(G\) ordering \(v_1, \ldots, v_n\).
Select some arbitrary vertex \(v_{q_1}\) and consider the walk
\(v_{q_1},\ldots,v_{q_{n}}\) around \(H\) that starts 
and ends at \(v_{q_1}\).

Let \(M\) and \(A\) be the two inputs to \omdci, where 
\(M = ([m], \gamma_M, \sigma_M)\) and \(A = ([n], \gamma_A, \sigma_A)\)
as described in
Section~\ref{subsec:instance}.

We start by constructing subsequences \(A'\) in \(A\) and \(M'\) in \(M\) 
before showing 
properties
\eqref{property:colorsaremakeacommonsubsequence} and
\eqref{property:setofcoveredcriticalinstructionsareequal} of Definition~\ref{def:gomdic}.
We start by considering the subset indices in \(M'\) and \(A'\)
in the \(Selection\) gadget which denote \(M'_{P}\) and \(A'_{P}\).
Let \(\Gamma'_{P} = \{P_{i, q_i}: i \in [N]\}\), then let
     \[ M'_{P} = \{i \in [m]: i \in \Gamma'_{P}\}, \quad A'_{P} = \{i \in [N]: i \in \Gamma'_{P}\}.\] 
Using Propositions~\ref{prop:p-to-s}, \ref{prop:s-to-e} and
\ref{prop:e-to-l-p} as a guide, we now consider the remaining
indices in \(A'\) and \(M'\).
Let
    \(\Gamma'_{S} = \{S_{i, q_i}: i \in [n]\}\)
and 
    \(\Gamma'_{T} = \{T_{i, q_i}: i \in [n]\}\).
The set of indices in \(M'\) and \(A'\) 
    with color in \(\cup_{i}\Gamma^{(S)}_{i}\) are respectively:
    \[I^{(M)}_{S} = \{i \in [m]: \gamma_M(i) \in \Gamma'_{S}\},\]
and
    \[I^{(A)}_{S} = \{i \in [N]: \gamma_A(i)\in \Gamma'_{S}\}.\]
The set of indices in \(M'\) and \(A'\) 
    with color in \(\cup_{i}\Gamma^{(T)}_{i}\) are respectively:
\[ I^{(M)}_{T} = \{i \in [m]: \gamma_M(i)\in \Gamma'_{T}\},\]
and
\begin{equation}\label{eq:a_t_def}
    I^{(A)}_{T} = \{i \in [N]: \gamma_A(i) = T_{i, q_i}
            \text{ and } \sigma_A(i) = p_{1 + (i \%n), q_{1 + (i\%n)}}\}.
\end{equation}
Now, let \(A'\) be the subsequence in \(\gamma_A\) defined by the sequence of indices \(I^{(A)}_{P} \cup I^{(A)}_{S} \cup I^{(A)}_{T}\),
and let \(M'\) be the subsequence in \(\gamma_M\) defined by the sequence of indices \( I^{(M)}_{P} \cup I^{(M)}_{S} \cup I^{(M)}_{T}\).

By Observations~\ref{obs:eachcolorappearsonceinM} and
\ref{obs:exactlyonecolorexceptedges}, each color in
\(\Gamma'_P\), \(\Gamma'_S\) appears only once in \(A\) and \(M\) 
thus appears only once in \(A'\) and \(M'\).
Further Observation~\ref{obs:eachcolorappearsonceinM} also implies each \(\Gamma'_E\)
appears once in \(M'\) and \(M\), and similarly \eqref{obs:endsoccurdegeachunique}
and \eqref{eq:a_t_def} imply the same for \(A\) and \(A'\).
Thus both \(A'\) and \(M'\) are subsequences of length \(3n\).
Note \(I^{(A)}_P\) and \(I^{(M)}_P\) each only use one index in each subgadget
of the \(MSelection\) and \(ASelection\); thus, the first \(n\) colors
of \(M'\) and \(A'\) are identical.
Similarly, we have \(S_{i, q_i}\) and \(E_{i, q_i}\) appearing in both \(M'\)
and \(A'\)  with \(S_{i,q_i}\) appearing in the \(PreLink_{i,q_i}\), which 
occurs before the \(PostLink_{i, q_i}\) containing \(E_{i,q_i}\).
As \(S_{i,q_i}\) and \(E_{i, q_i}\) are the only two characters appearing 
in the \(Linking_{q_i}\) subgadget and the \(MLinkin_{q_i}\) and \(ALinking_{q_i}\)
are concatenated together in the same order in both \(A\) and \(M\),
the last \(2n\) of \(A'\) and \(M'\) must be identical as well.
Therefore \(A'\) and \(M'\) must be a common subsequence of \(A\) and \(M\)
satisfying Property \ref{property:colorsaremakeacommonsubsequence} of Definition~\ref{def:gomdic}.

We now proceed to show \eqref{property:setofcoveredcriticalinstructionsareequal}
holds for \(A'\) and \(M'\).
Let \(\Sigma_A\) and \(\Sigma_M\) be the set of instructions in \(\Sigma\)
    that appear at least once in \(A'\) and \(M'\) respectively.
Note for \(i \in [n]\), there is an index with critical character
\(P_{i, q_i}, S_{i, q_i}, T_{i, q_i}\) in \(M'\) thus applying
Observation~\ref{obs:eachcolorappearsonceinM} tells us:
\begin{enumerate}
    \item \(p_{i, q_i} \in \Sigma_M\) as \(P_{i, q_i}\) 
        appears in \eqref{eq:mselwi},
    \item \(s_{i, q_i} \in \Sigma_M\) 
        as \(S_{i, q_i}\)  appears in
        \eqref{eq:mprelink-with-ij-upto-n-1},
    \item \(t_{i, q_i} \in \Sigma_M\) as  \(T_{i,q_i}\) appears 
        in \eqref{eq:mpostlink-with-ij}.
\end{enumerate}
We further have
by consequence of Observations~\ref{obs:exactlyonecolorexceptedges} and \ref{obs:exactlyoneciexceptpicksci}:
\begin{enumerate}
    \item \(p_{i, q_i} \in \Sigma_A\) as there is an index \(i'\) in \eqref{eq:apostlink-with-ij}
        where \(\sigma_A(i') = p_{i,q_i}\) and \(\gamma_A(i')=E_{i-1,q_i}\)
        for \(i \in [2,n]\) and \(\gamma_A(i')=E_{n,q_1}\) when
        \(i = n\).
    \item \(s_{i,q_i} \in \Sigma_A\) as \(P_{i, q_i}\) appears
        in \(A'\) and \eqref{eq:aselwi}.
    \item \(t_{i,q_i} \in \Sigma_A\) as \(S_{i, q_i}\) appears in \(A'\)
        and \eqref{eq:apostlink-with-ij}.
\end{enumerate}
As we have now considered every index in \(M'\), 
we have \(\Sigma_M \subseteq \Sigma_A\).
By Observation~\ref{obs:eachcolorappearsonceinM} we have 
\(|\Sigma_M| = |M'|\) thus
\(\Sigma_M = \Sigma_A\) as \(|M'| = |A'| > |\Sigma_A|\).   
This shows property~\ref{property:setofcoveredcriticalinstructionsareequal} in
Definition~\ref{def:gomdic}, which completes the forward direction.

We now show the reverse direction.
Let \(M'\) and \(A'\) be an OMDCI solution for \(M\) and \(A\).
As \(k>0\), there exists some \(c \in \Gamma\) that is a color of an
index in both \(A'\) and \(M'\).
Suppose \(c \not \in \cup_i^n \Gamma_i^{(P)}\) then 
Propositions~\ref{prop:s-to-e} and \ref{prop:e-to-l-p}
imply that there exists some \(c' \in \cup_i^n \Gamma_i^{(P)}\) that
is the color of an index in both \(M'\) and \(A'\).
Thus, we may always select a \(c \in \cup_i^n \Gamma_i^{(P)}\) that
is the color of an index in both \(A'\) and \(M'\).
Furthermore, as a consequence of Proposition~\ref{prop:one-per-selection}
and Lemma~\ref{prop:p-to-p}, we have \(n\) colors 
\(P_{1, h_1}, \ldots ,P_{n, h_n}\) that are the colors of indices in both \(M'\) and \(A'\) as well as colors 
\(S_{1,h_1}, \ldots S_{n, h_{n}}\) that are also the 
color of indices in both \(M'\) and \(A'\) where
\(v_{h_n} \sim_E v_{h_0}\) and \(v_{h_i} \sim v_{h_{i+1}}\) for 
\(i \in [n-1]\).
Now all that is left to show is that \(h_i \neq  h_j\) if \(i \neq j\).
We proceed by contradiction and assume without loss of generality 
\(i < j\) and \(h_i = h_j\).
Note then \(S_{i, h_i}\) and \(S_{j, h_j} = S_{j,h_i}\)
both appear in \(M'\) and \(A'\)
which is a contradiction by Proposition~\ref{prop:no_l_repeats}.
Therefore we have found \(n\) unique vertices such that
\(v_1 \sim_E v_2 \sim_E \ldots \sim_E v_n \sim_E v_1\) thus
\(G\) has a Hamiltonian cycle.
\end{proof}

\end{document}